\documentclass[11pt]{article} 

\usepackage[utf8]{inputenc} 
\usepackage{amsthm}

\usepackage{geometry} 
\usepackage{graphicx} 

\usepackage{booktabs} 
\usepackage{array} 
\usepackage{paralist} 
\usepackage{verbatim} 
\usepackage{subfig} 
\usepackage{amsfonts}
\usepackage{amssymb}
\usepackage[cmex10]{amsmath}
\usepackage{mathtools}
\usepackage{breqn}

\usepackage{fancyhdr} 
\usepackage{sectsty}
\allsectionsfont{\sffamily\mdseries\upshape} 

\usepackage[nottoc,notlof,notlot]{tocbibind} 
\usepackage[titles,subfigure]{tocloft} 

\newtheorem{ar}{A}
\newtheorem{ma}{M}
\newtheorem{ch}{C}
\newtheorem{theorem}{Theorem}

\newtheorem{corollary}{Corollary}

\title{Self-consistency of voting implies majority vote}
\author{Artur Poplawski}

\begin{document}
\maketitle
\begin{abstract}
Paper develops axiomatic characterization of the family of majority vote rules in the way alternative to characterization of the majority vote given 
in paper of Kenneth O May in the 1952. This, similar but different, axiomatics focuses on the consistency of the voting procedure. 
Both approaches are compared. Relation to famous Kenneth J. Arrow's Impossibility Theorem is also discussed.
\end{abstract}
\section{Introduction}
Large number of studies in the fields attributed to Game Theory, mathematical political science  or mathematical theory of welfare 
is devoted to problem of procedure of collective choice (or simply voting). The most famous and celebrated result in this area is 
so called Arrow's Impossibility Theorem published by Kenneth J, Arrow in his article \cite{Ar}. 
Following style of Arrow's work and using the same axiomatic approach, Kenneth O. May published theorem characterizing majority vote
by the set of the postulates or conditions that on the general voting procedure: i.e showing that any procedure satisfynig these postulates 
is in fact a majority voting.
In this short paper we will give different characterization of the broader class of voting procedures each being some version of the 
majority vote. In the next section we will introduce notation and refer Arrow's and May results. Then, next section will contain 
main result and its proof. In last section we  compare proposed axiomatization with this of May, discuss to what extent it generalizes it
and present concepts for some further studies. 

\section{Arrow's and May result}
Let's fix notation. We will use capital letters ($A$, $V$ etc.) to denote sets, lower case letter ($x$, $v$, etc.) for variables 
denoting elements of sets. 
Cartesian products of set $A_{v}$ indexed by set $V$ vill be denoted by $\prod_{v \in V}A_{v}$. In the most common context we will have
$A_{v}=A$ for some $A$ and we will write simply $\prod_{v \in V}A$. Let us note, that the elements of this last set are all functions
$V \rightarrow A$ (so, using other common notation $\prod_{v \in V}A = A^{V}$. We will denote cardinality of some set $X$ by the $|X|$. 

In \cite{Ar} Arrow models situation where finite set of agents $V$ posses some preferences related to finite set of possible choices 
(alternatives) $A$. Voter's $v$ preference is weak order relation $R_{v}$  on the set $A$ 
\footnote{Weak order  set $A$ is a relation 
$R\subset A \times A$such that it is 1) reflective: $\forall x: x \in \implies (x,x) \in A $ 2) transitive: 
$\forall x,y,z \in A: (x,y) \in R, (y,z) \in R \implies (x,z) \in R$}.
Some voting procedure, or social choice leads to weak order $R$ on the set $A$ which, to some extent should reflect the preferences of the 
voters. Arrow considers following five conditions satisfying of which one may require from the fair voting procedure.

\begin{ar}\label{ar1}
Procedure leads to function $f:\prod_{v \in V}Ord_{A}\rightarrow Ord_{A}$ 
\end{ar} 
This condition just states that procedure is \emph{sound} - it depends only on the individual choices and is deterministic. It excludes the procedures like drawing. 

\begin{ar}\label{ar2}

$\forall x, y \in \prod_{v \in V}Ord_{A}, \forall a, a' \in A: ((a,a') \in f(x) \land  v' \in V \and (\forall v \in V-\{v'\} x_{v}=y_{v}) \land (\forall s, t \in A-\{a, a'\}: (s,t) \in x_{v'} \iff (s,t) \in y_{v'}) \land (a', a) \in x_{v'} \land (a,a') \in y_{v'} \implies (a, a') \in f(y) $ 

\end{ar} 
This somewhat complicated statement describe \emph{nonnegative responsiveness} to the choice of individuals, and may be stated
in less formal way as: if choice of the single voter changes with respect to two alternatives $a$ and $a'$ in the favor of the social choice for these alternatives, outcome of the social choice regarding these alternatives should not change.  

\begin{ar}\label{ar3}
$\forall a, a' \in A, \forall x, y \in \prod_{v \in V}Ord_{A}, \forall v \in V: (a,a') \in x_{v} \iff (a',a) \in y_{v} \implies (a,a') \in f(x) \iff (a,a') \in f(y)$
\end{ar}
This condition is called \emph{indifference to irrelevant choices}. It says, less formally, that if voters change their preferences but preserve preferences regarding some two alternatives, the prefernece between these alternatives in the social choice should also stay unchanged.   

\begin{ar}\label{ar4}
$\forall a, a' \in A \exists x, y \in \prod_{v \in V}Ord_{A}: (a,a') \in f(x) \land (a',a) \in f(y)$
\end{ar}
This means, that function is \emph{not imposed} in the terminology of Arrow. Every choice of social preferences is a priori possible as an outcome, there is no taboo.

\begin{ar}\label{ar5}
$\forall a, a' \in A \forall v \in V \exists x \in \prod_{v \in V}Ord_{A}: (a,a') \in x_{v} \land (a,a') \notin f(x)$ 
\end{ar}
In other words, there is no dictator among voters - every vote can be overthrown. Arrow calls this property \emph{no dictatorship}

The thesis of the Arrow is that there is no system satisfying the set of postulates A \ref{ar1}-A \ref{ar5}. It also often expressed in the
literature in equivalent but even more striking form:

\begin{theorem}[Arrow's Impossibility Theorem]\label{ait}
Every voting procedure leading to the function satisfying A \ref{ar1}= A \ref{ar4}, there is voter $v$ who is the dictator. 
\end{theorem} 

By the way, it is the existence of the dictator in some cases what prevented us from using the Cartesian power $A^{n}$ in statements 
of the conditions - in Arrow's setup the individuality may matter.

Besides the natural philosophical discussion of the consequences of this theorem for theory of democracy and practical questions
in what sense  rules applied in the democratic institution violate and which of the Arrow's condition, his publication triggered 
further studies on formal properties of voting procedures. As one can perceive Arrow's result as counterintuitive, in 1952 
Kenneth O. May published the paper \cite{May} were he stated positive and much more intuitive result. He considered  following 
situation, and entity he called \emph{group decision function}.
We have set of voters $V$ and fix the set of three alternatives: $A=\{-1,0,1\}$.  Outcome of the procedure is the choice of the oner of the alternative.
$-1$ and $1$ are "real" choices, $0$ is treated as as $no-choice$, so skipping the voting from the perspective of the voter and tie
(whatever tie would mean for an arbitrary procedure) so lack of decision when treated as an outcome.
In such formalism May postulates following four conditions.

\begin{ma}\label{ma1}
Procedure leads to function $f:\prod_{v \in V}A\rightarrow A$. We will use $f$ to name this function is subsequent conditions.
\end{ma}
It is similar to Arrows \ref{ar1} and have the same meaning of soundness.

\begin{ma}\label{ma2}
$f$ is symmetric, so $\forall x, y \in \prod_{v \in V}A, \forall v, w \in V: (\forall z \in V-\{v,w\} x_{z}=y_{z} \land x_{v}=y_{w} \land y_{v}= y_{w}) \implies f(x)=f(y)$
\end{ma}
Here we are a little bit more pedantic than May's paper, although May's condition may be easily deduced to  be equivalent to this  one. 
Ma \ref{ma2} may be also interpreted as assumption that $f$ does depend only on the number of voters selecting particular choice. It also 
implies, that procedure does not have dictatorship.

\begin{ma}\label{ma3}
$\forall x, y' \in \prod_{v \in V}A: (\forall v \in V:  x_{v} = -y_{v}) \implies f(x) = -f(y)$ 
\end{ma}
This condition is called \emph{neutrality} by May. It is coveniently formulated thanks to the special "algebraic" structure of the 
in the set of alternatives. 

\begin{ma}\label{ma4}
$\forall x, y \in \prod_{v \in V}A \exists v \in V : (\forall w \in V-\{v\} x_{w}=y_{w}) \land x_{v} \neq 0 \land x_{v}=-y_{v} \land (f(x)=0 \lor f(x)=y_{v}) \implies f(y)=y_{v}$
\end{ma}
May calls this condition \emph{positive responsiveness}. Intuitive meaning of this condition is, the "change of the mind" of voter in favor of the 
collective choice or in the case of tie leads to the same choice in former case and to the switching the collective choice according to the decision
of voter in the other. 

May proves in elementary way (proof is much simpler than known proofs of Arrow's theorem) that the only $f$ satisfying this condition is the "majority vote", 
which in his setting  can be described as:
 
\[
f(x) = \begin{cases}
1 &\text{if } \sum_{v\in V}x_{v} > 0\\
0 &\text{if } \sum_{v\in V}x_{v} = 0\\
-1 &\text{if } \sum_{v\in V}x_{v} < 0\\
\end{cases}
\]

\section{Main results.}
Here we formulate result similar to May's, assuming however, instead of May's conditions Ma \ref{ma2} - Ma \ref{ma4} some structural rule concerning 
the consistency of the voting procedure among different setting.
Let's note, that both, Arrow's and May's conditions are expressed as a postulates to a procedure applicable to some given set of voters. I.e. 
there is no a priori relation between the considered procedures of voting for different number of voters. 
It is, perhaps, more natural to formulate such rules in the way that is applicable to all possible finite set. 
Although this distinction is not essential for expressing May's and Arrow's results it is hidden assumption.
We will formulate three postulates for the set of voters $V$, set of alternatives $A$ and "consistent" family of choice procedures, where members of the family are the procedures for the set of voters of different cardinality. We will assume, that in $A$ there is a special alternative $\bot$ which will have interpretation of the tie. We also introduce simple weak ordering relation $R$ in $A$, which will be defined as a $R=\{(\bot, x)| x \in A\}$

\begin{ch}\label{ch1}
Procedure leads to functions $f_{i}: \prod_{v \in V}A  \rightarrow A$ for $|V| = i$
\end{ch}
Although we consider the family of functions, $\{f_{i}\}_{i \in \mathbb{N}}$, whenever it will not lead to the confusion 
we will use unindexed symbol $f$ assuming that $i$ is equal to the cardinality of $|V|$ of currently considered $V$.

\begin{ch}\label{ch2}
For set of voters $V$ any $x \in \prod_{v \in V}A$ and any permutation $\pi:A \rightarrow A$ such that $\pi(\bot)=\bot$ we have 
$f(\pi \circ x ) = \pi(f(x))$  
\end{ch}

\begin{ch}\label{ch3}
For set of voters $V$ any $x \in \prod_{v \in V}A$ and any permutation $\pi:V \rightarrow V$ such that we have 
$f(x \circ \pi) = f(x)$  
\end{ch}

Both conditions together form condition similar to the Ma \label{ma2} and expresses symmetry so, in fact, the property, 
that $f$ depends only on number of voters not on the particular voter.  
 
\begin{ch}\label{ch4}
For set of voters $V$ any $x \in \prod_{v \in V}A$ if $w \notin V$ and $x' \in \prod_{v \in V\cup\{w\}}A$  such that 
$\forall v \in V, x_{v}=x'_{v}$ and $x'_{w} = \bot$ we have $f(x)=f(x')$
\end{ch}

This may be considered as a implicit definition of the skipping (if interpreted as voters decision) 
or tie (if interpreted as social choice), and distinguish tie among the possible alternatives. 

\begin{ch}\label{ch5}
For set of voters $V$ any $x \in \prod_{v \in V}A$ if $w \notin V$ and $x' \in \prod_{v \in V\cup\{w\}}A$  such that 
$\forall v \in V, x_{v}=x'_{v}$ and $x'_{w} = f(x)$ we have $f(x)=f(x')$
\end{ch}
This is the condition, that we can call \emph{consistency}. Intuitively it states, that decision procedure is such, that
whenever one extends the set of voters with one new voter with preference being the same social choice of the initial set, 
the social choice of the new set must remain the same.
We introduce notation, for $a \in A$ and $x \in \prod_{v \in V}A$  $\phi_{a}(x) = |\{v| x_{v}=a\}|$
Main result of the work is following:
\begin{theorem} \label{main}
For set of alternatives $A$ with distinguished alternative $\bot$ and for family of sets $\{V_{i}\}_{i\in\mathbb{N}}$, any procedure satisfying the C \ref{ch1} - C \ref {ch4} leads to social choice function $f$ such that: $f(x) \neq \bot \implies \forall a \in A -\{f(x), \bot \}\phi_{f(x)}(x)>\phi_{a}(x)$
\end{theorem}
Intuitively speaking, result says, that procedure assuming our postulates leads to the social choice functions which, if returns conclusive
result, the result must be choice of the absolute majority.

\begin{proof}

Let's take $a \in \prod_{v \in V}A$. If there is $x \in A$ such that $x \neq \bot$ and for all $y \in A$ such that $y \neq x$
and $y \neq \bot$ we have $\phi_{a}(x) > \phi_{a}(y)$, than $f(a) = x$ or $f(a) = \bot$.
Indeed, if $f(a) = y$ for some $y \neq x$ we know from C \ref{ch4} that enlarging $V$ by the single element $v'$, in such a way, that
preference of this new voter is $y$ (we will receive tuple $a'=a \cup (v',y)$ for which $f(a') = y$.
Repeating the operation $k = \phi_{a}(x) - \phi_{a}(y)$ times, we receive $a^{\star} \in \prod_{v \in V'}A$ where $V'=V \cup \{v{1}, \ldots, v_{k}\}$, such that 
$a^{\star}(v)=a(v)$ for $v \in V$, $a^{\star}(v_{i})=y$ for $i=1,\ldots,k$ and $\phi_{a^{\star}}(x) = \phi_{a^{\star}}(y)$.
We also have from C \ref{ch4} that $f(a^{\star}) = y$. 

Now, let's define two permutations. First is  $\pi_{A}:A \rightarrow A$:
\[
\pi_{A}(t) = \begin{cases}
x &\text{if } t = y\\
y &\text{if } t = x\\
t &\text{otherwise}\\
\end{cases}
\]
Obviously under our assumptions we have $\pi_{A}(\bot)=\bot$.
Let's take any permutation $\chi: (a^{\star})^{-1}(x) \rightarrow (a^{\star})^{-1}(y)$. Such a permutation exists, because of equal cardinalities of both preimages.
We define second permutation $\pi_{V'}:V' \rightarrow V'$ by:
\[
\pi_{V}(t) = \begin{cases}
\chi(x) &\text{if } t \in (a^{\star})^{-1}(x)\\
\chi^{-1}(y) &\text{if } t \in (a^{\star})^{-1}(y)\\
t &\text{otherwise}\\
\end{cases}
\]

Under our assumptions about $a$, properties of $a^{\star}$ and definitions we have just given, we have:
$y = f(a^{\star}) = f(a^{\star} \circ \pi_{V}) = f(\pi_{A} \circ a^{\star}) = \pi_{A}(f(a^{\star}) = \pi_{A}(y)=x$
what is a contradiction.  
\end{proof}

From the theorem we have immediate corollary:
\begin{corollary}\label{col1}
Assuming $f$ satisfies C \ref{ch1}-\ref{ch5}, if $a \in \prod_{v \in V}A$ is such, that there exists $x, y \in A$ such that $x \neq y$,$ \phi_{a}(x) = \phi_{a}(y)$ and for 
all $z \in V$ $\phi_{a}(x) \geq \phi_{a}(z)$ we have $f(a) = \bot$
\end{corollary}
\begin{proof}
Let's assume that we $f(a) = s \neq \bot$. Either $s$ is such that $\phi_{a}(z)=\phi_{a}(x)$ than we can construct the permutations as n the proof of the main theorem 
taking two elements: $s$ and one of the elements in $\{x,y\}-\{s\}$ and obtaining the contradiction. Otherwise, again we can reason as in the proof of main theorem 
obtaining $a^{\star}$ such that $f(a^{\star})=s$, $\phi_{a^{\star}}(s) =  \phi_{a^{\star}}(x) = \phi_{a^{\star}}(y)$ and $\phi_{a^{\star}}(s) \geq \phi_{a^{\star}}(t)$ 
for all $t \in dom(a^{\star}))$  
\end{proof}

From what we have just proven it is clear, that procedure satisfying C \ref{ch1}-\ref{ch5} are majority voting with some situations of tie - so voting is not conclusive.
There are two types of situations of tie. Some are unavoidable, so are simply consequences of the assumptions about the the procedure and will be a tie in every rule - these 
are described by the corollary \ref{col1}. Others may be called singular and it is them which differentiate between procedures. For each procedure however, the 
singular existence of some singular tie may imply existence of the other singular ties. We have obvious consequence of the C \ref{ch5}  
\begin{corollary}\label{col2}
For each $a \in \prod_{v \in V}A$ if $f(a)=\bot$ $a \neq b \cup (v, x)$ where $b \in V-\{v\}$ and $f(b) = x$ and $x \neq \bot$
\end{corollary}
In particular, if $f(a) = \bot$ and $x \neq \bot$ is such that $\phi_{a}(x) > \phi_{a}(z)$ for $z \neq x$ we have $f(a') = \bot$ where $a'=a|_{V -\{v\}}$ for $v$ such that $a(v)=x$.

Now let's characterize "pure majority" vote as vote without the singular ties. For established $V$ and $A$ there is obviously only single pure majority vote procedure.
It van be also characterized in other form. If we assume, under the set of alternatives $A$ that voting procedure is defined on the finite subsets of the universum of 
voters $W$, we can introduce partial order among the procedures, by specifying $f < g$ iff for each $V \in Fin(W)$ and each $a \in \prod_{v \in V}A$ $f(a) = \bot$ or $f(a) = g(a)$.
Pure majority vote will be the (single) maximal element according to this order.  

\section{Comments and future work}
Characterization of the majority rules by C \ref{ch5} is different in spirit comparing to the May's condition. Also conditions here cover more general and at the same time 
more delicate situation of multiple choices. That is partially the reason while conditions does not look so elegant: May considering only the dual choice and tie, introduced 
some algebraic structure in the set of alternatives what allowed him to express majority voting in concise form. Also his conditions of the breaking the ties arise in natural way
while here, to get pure majority rule we used lack of the singular ties or maximum of special order. It is however worth to note, that if one introduces additional condition, 
exactly like a May's, namely:
\begin{ch}
For each $a \in \prod_{v \in V}A$ we have $f(a) = \bot$ implies existence of the $b \in  \prod_{v \in V \cup \{v'\}}A$ where $v' \notin V$ and $a=b|_{V}$ such that
$f(b) \neq \bot$ 
\end{ch}
the only procedure satisfying all conditions would be pure majority rule.
However, some singular ties are are quite common in voting procedures we can meet in nature. Typical example is a quorum rule, where $f=\{f_{i}\}_{i \in \mathbb{N}}$ is such that
$f_{i} \equiv \bot$ for $i < N$ for some $N$ and $f_{i}$ being the majority rule for $i \geq N$. Other examples are voting procedures where to choose appropriate fraction
of votes must be caat on choice. E.g. such a rule would be $f(a) = x \neq \bot$ if and only if $ \phi_{a}(x) > \frac{1}{2}|a|$ or, as a different example 
$ \phi_{a}(x) > \frac{1}{2}(|a| -\phi_{a}(\bot)$

It would be interesting to find simple and natural axiom to be added to conditions C in order to eliminate "undesirable" singular ties but still 
cover the singular ties leading to quorum or "strong majorities" like in examples above.

Looking at the proof of the Arrow theorem and it relationship to some abstract mathematical constructions (e.g. \cite{Ab}), one can also ask
if, although much simpler, May's theorem and Theorem \ref{main} have any interesting intepretations in similar spirit.


\begin{thebibliography}{1}
\bibitem{Ar}
Kenneth J. Arrow  
\emph{A Difficulty in the Concept of Social Welfare}, The Journal of Political Economy, Vol. 58, No. 4. (Aug., 1950), pp. 328-346.
\bibitem{May}
Kenneth O. May 
\emph{A Set of Independent Necessary and Sufficient Conditions for Simple Majority Decision}, Econometrica, Vol. 20, No. 4 (Oct., 1952), pp. 680-684, pp. 328-346.
\bibitem{Ab}
Samson Abramsky 
\emph{Arrow's Theorem by Arrow Theory}, arXiv:1401.4585v2
\bibitem{You}
H.Peyton Young 
\emph{Equity: In Theory and Practice}, 1994 Princeton University Press

\end{thebibliography}
\end{document}